\documentclass{eptcs}
\usepackage[utf8]{inputenc}
\usepackage{graphicx}
\usepackage{amssymb}
\usepackage{amsmath}
\usepackage{amsthm}
\usepackage{times} 
\usepackage{color}
\usepackage{bm}
\usepackage{soul}  
\usepackage{hyperref} 
\usepackage{graphics}

\usepackage{booktabs}
\usepackage[linesnumbered,ruled,vlined]{algorithm2e}
\usepackage{subcaption}

\newcommand{\tuple}[1]{\langle #1 \rangle}
\def\eqdef{\stackrel{\rm def}{=}}

\usepackage[inline]{enumitem}
\usepackage{tikz}
\usetikzlibrary{trees}
\usepackage{stackrel}
\usepackage{breakurl}

\newcommand{\GTROIS}{\ensuremath{G_3}}
\newcommand{\QGTROIS}{\ensuremath{QG_3}}
\newcommand{\BINDER}{\ensuremath{\overline{Q}}}
\newcommand{\EMPTYBINDER}{\ensuremath{\varepsilon}}
\newcommand{\EMPTY}{\ensuremath{\bm{\varepsilon}}}
\newcommand{\FAUX}{\ensuremath{0}}
\newcommand{\UNDEMI}{\ensuremath{\frac{1}{2}}}
\newcommand{\VRAI}{\ensuremath{1}}
\newcommand{\eQBFpol}[3]{\ensuremath{#2;#3}}
\newcommand{\uQBFpol}[3]{\ensuremath{[\FAUX \to #2 \mid \VRAI \to #3]}}

\newcommand{\eGpol}[3]{\ensuremath{#2;#3}}
\newcommand{\uGpol}[4]{\ensuremath{[\FAUX \to #2 \mid \UNDEMI \to #3 \mid \VRAI \to #4]}}

\newcommand{\prop}{\ensuremath{\mathtt{VAR}}}
\newcommand{\HT}{\ensuremath{\textrm{HT}}}
\newcommand{\TRIVAL}{\ensuremath{\bm{m}}}
\newcommand{\trivaluation}[2]{\ensuremath{\bm{#2}(#1)}}
\newcommand{\trival}[1]{\trivaluation{#1}{m}}

\newcommand{\assign}[2]{\ensuremath{\TRIVAL^#1_#2}}
\newcommand{\assigns}[3]{\ensuremath{\TRIVAL^#1_#2(#3)}}

\newcommand{\inters}[3]{\ensuremath{\TRIVAL^#1_#2(#3)}}

\newcommand{\crisp}[2]{\ensuremath{\mathtt{crisp}_{#1}^{#2}}}
\def\QEM{\ensuremath{\mathit{QEM}}}
\newcommand{\qem}[3]{\ensuremath{\QEM(#1,#2,#3)}}
\def\MCEXISTS{\ensuremath{\mathit{mc}_{\exists}}}
\newcommand{\mcexists}[4]{\ensuremath{\MCEXISTS(#1,#2,#3,#4)}}
\def\MCFORALL{\ensuremath{\mathit{mc}_{\forall}}}
\newcommand{\mcforall}[4]{\ensuremath{\MCFORALL(#1,#2,#3,#4)}}

\newcommand{\egalite}[1]{\ensuremath{\triangleq_{#1}}}

\newtheorem{theorem}{Theorem}
\newtheorem{lemma}[theorem]{Lemma}
\newtheorem{corollary}[theorem]{Corollary}
\newtheorem{proposition}[theorem]{Proposition}

\newtheorem{definition}{Definition}
\newtheorem{example}{Example}

\begin{document}

\title{Extracting Policies from Quantified Answer Set Programs}
\author{
	Mart\'{\i}n Di\'eguez
	\institute{University of Angers, France}
	\email{martin.dieguezlodeiro@univ-angers.fr}
	\and 
	Igor St\'ephan
	\institute{University of Angers, France}
	\email{igor.stephan@univ-angers.fr}
}

\newcommand{\titlerunning}{Extracting Policies from Quantified Answer Set Programs}
\newcommand{\authorrunning}{M. Di\'eguez and I. St\'ephan}
\maketitle
\begin{abstract}
Quantified Answer Set Programming (QASP) extends Answer Set Programming (ASP) by allowing quantification over propositional variables, similar to Quantified Boolean Formulas (QBF). In this paper, we interpret models of QASP formulas in terms of \emph{policies}, which represent decision-making strategies that determine how existentially quantified variables should be assigned, given the conditions set by universally quantified variables. As a main contribution, we present an algorithm for policy extraction under QASP semantics, inspired by the Equilibrium Logic semantics for general ASP theories.
\end{abstract}

\section{Introduction}\label{sec:intro}
Having its roots in logic programming, \emph{Answer Set Programming}~\cite{breitr11a} (ASP) is a formalism for nonmonotonic reasoning (NMR) and Knowledge Representation (KR) that solves complex problems involving combinatorial search, optimisation, and reasoning under uncertainty. In ASP, problems are encoded as a set of logical rules, where the corresponding models (answer sets) represent solutions. Besides the traditional \emph{stable models} semantics, there exist alternative interpretations of logic programs under answer set semantics~\cite{Lifschitz2010}, including Pearce's \emph{Equilibrium Logic}~\cite{pearce06a} (EL), which uses a monotonic basis and a minimality condition to induce nonmonotonicity. EL has been extended to enhance expressiveness for problems like temporal or epistemic reasoning.
For propositional theories, the satisfiability problem in ASP lies within the second level of the polynomial hierarchy ($\Sigma_2^p$)~\cite{petowo01a}. Approaches like stable-unstable semantics~\cite{bojata16a} extend ASP to solve problems beyond this level by integrating logic programs as oracles.

A second approach to problem solving is the ASP(Q) system\footnote{Available at \url{https://www.mat.unical.it/ricca/downloads/qasp-0.1.2.jar}}~\cite{AMENDOLA_RICCA_TRUSZCZYNSKI_2019}, which extends the ASP syntax with existential and universal quantifiers that range over the stable models of a program rather than over atoms. In terms of applications, ASP(Q) has been employed in representing and solving graph problems~\cite{AmendolaCRT22}, as well as in epistemic logic programming~\cite{0001M23} and argumentation~\cite{DBLP:conf/iclp/000124}.

\noindent In the classical setting, \emph{Quantified Boolean Formulas}~\cite[Chapter 31]{handbookSAT} (QBFs) extend propositional logic with second-order quantification, allowing for the expression of problems beyond the complexity class $\Sigma_2^p$, since QBF satisfiability is PSPACE-complete~\cite{Stockmeyer_Meyer_STOC_73}. 
\emph{Quantified Answer Set Programming} (QASP)~\cite{ua8046,falaroscso21a} similarly extends ASP by incorporating propositional quantifiers, analogous to QBF. As in QBF, both existential and universal quantifiers apply to the truth values of propositional atoms. However, the key difference lies in the semantics: in QASP, the truth values assigned to atoms must be consistent with the stable models of the program.
There exist two primary semantics for QASP programs. In the semantics proposed by Fandinno et al.~\cite{falaroscso21a}, an atom $p$ is forced to be true in a stable model by adding the formula $\neg\neg p$ to the context. In contrast, in the semantics defined by Stephan~\cite{ua8046}, $p$ is made true simply by adding $p$ to the context. Both semantics behave identically when determining whether an atom is false in a stable model.
In terms of implementation, the semantics from~\cite{falaroscso21a} are realized in a tool called \texttt{qasp2qbf}\footnote{Available at~\url{https://github.com/potassco/qasp2qbf}.}, which has been used to solve problems in conformant planning. Furthermore, it is shown in~\cite[Appendix B]{falaroscso21a} that ASP(Q) and QASP are intertranslatable.

In QASP, solutions to quantified logic programs can be interpreted as two-player games with complete information. The first player attempts to construct stable models of the ASP program, and the program is considered a positive instance of the QASP problem if and only if a winning strategy exists for the first player. The tool \texttt{qasp2qbf} reduces a QASP program to a QBF, which is then passed to a QBF solver to determine satisfiability. If the QBF is satisfiable, the solver provides an assignment for the initial block of existential quantifiers. This means that for QASP programs beginning with a universal quantifier, the solver can only report satisfiability status, without producing a concrete strategy. At present, the tool does not support the construction of such two-player games. In contrast, the approach in~\cite{ua8046} replaces existential quantifiers with Skolem functions~\cite{Benedetti_CADE_05}, effectively simulating the existential player's choices. This method is capable of computing all possible two-player games that satisfy the QASP program.

In this paper, we examine~\cite{falaroscso21a} and~\cite{ua8046} semantics from the perspective of Equilibrium Logic (EL). We introduce the following restrictions: we do not consider a specific logic programming language; instead, we translate any logic programming rule $H \leftarrow B$ into the implication $B \to H$ in propositional logic. Default negation (\texttt{not}) is replaced by the negation symbol ($\neg$) in propositional logic, and propositional formulas can be arbitrarily nested. Additionally, we treat any arbitrary propositional formula\footnote{Since EL allows interpreting arbitrary (propositional) formulas.} as the matrix of a Quantified Boolean Formula (QBF).

The first contribution of this paper is to demonstrate that the two semantics are not equivalent. The second contribution is the extension of the notion of policy from the classical case to the quantified propositional G\"odel logic \QGTROIS{}~\cite{BaazCZ00}, which provides a monotonic basis for selecting minimal policies. The third contribution is the development of an algorithm that takes a QBF as input and computes the set of policies satisfying it under the semantics of~\cite{falaroscso21a}. This algorithm combines QBF and \QGTROIS{}-policies during execution.

This paper is organised as follows: Section~\ref{sec:qbf} introduces QBFs and the concept of policy. Section~\ref{sec:eqlogic} formulates EL in terms of the G\"odel logic \GTROIS{} plus a minimality condition, proving several interesting properties. Section~\ref{sec:qasp} presents the two QASP semantics used in this paper. Section~\ref{sec:policies} introduces the semantics for propositional quantifiers in \GTROIS{} and the concept of \QGTROIS-policy. In Section~\ref{sec:algorithmic}, we present an algorithm, based on EL, for extracting policies from QBFs under QASP semantics. The paper concludes with Section~\ref{sec:conclusions}, discussing future work.
 
\section{Background}
\subsection{QBFs and QBFs Policies}\label{sec:qbf}
Let $\prop$ be a (nonempty) set of atoms. A QBF (in \emph{prenex normal form}) is generated by the grammar:
\begin{eqnarray}
    \varphi,\psi &::=& p \mid \bot \mid (\varphi \wedge \psi) \mid (\varphi \vee \psi) \mid (\varphi \to \psi) \label{grammar:propositional}\\
    F &::=&  \varphi \mid  \forall x\; F \mid \exists x\; F, \label{grammar:quantifiers}
\end{eqnarray}
\noindent where $p, x \in \prop$. We will usually use the common rules for the elimination of parentheses, when necessary. The first part of the grammar~\eqref{grammar:propositional} generates propositional formulas while~\eqref{grammar:quantifiers} generates the prefix of quantifiers.
Given a QBF $Q_1x_1 \cdots Q_nx_n \varphi$ where $Q_i\in \lbrace \forall,\exists\rbrace$ and $x_i\in \prop$, for all $1\le i \le n$ as $\overline{Q}\varphi$, we will name the sequence $Q_1x_1 \cdots Q_nx_n$ the \emph{binder} and we will name the quantifier-free formula $\varphi$ the \emph{matrix}. 
Symbol \EMPTYBINDER\ denotes the empty binder.
For the semantics we define a \emph{classical interpretation} as a mapping $\bm{m}: \prop \mapsto  \lbrace \FAUX, \VRAI\rbrace$. For a given propositional variable $p$, $\trival{p}=\FAUX$ means that $p$ is false; $\trival{p}=\VRAI$ means that $p$ is true.
An interpretation $\bm{m}$ is said to be \emph{total} if it is defined for all $p\in \prop$ and \emph{partial} otherwise. If not specified explicitly, we will assume that $\bm{m}$ is total. 
Given an interpretation $\bm{m}$, $x \in \prop$ and $i\in \lbrace \FAUX,\VRAI\rbrace$, we define the \emph{update} of $x$ in $\bm{m}$ as 
\begin{equation}
\assigns{x}{i}{p} \eqdef \left\{ 
				\begin{array}{ll} 
				i & \hbox{if } p = x\\
				\trival{p} & \hbox{otherwise.}
				\end{array}\right. \label{eq:update}
\end{equation}    
Moreover, by \EMPTY\ we denote the (partial) interpretation such that for all $x\in \prop$, $\EMPTY(x)$ is not defined. Any interpretation $\bm{m}$ can be extended to any formula $\varphi$ by means of the following satisfaction relation:
\begin{eqnarray*}
 \trival{\bot}= \FAUX & \trival{\varphi \wedge \psi} =\min\lbrace\trival{\varphi},\trival{\psi}\rbrace & \trival{\varphi \vee \psi} =\max\lbrace\trival{\varphi},\trival{\psi}\rbrace \\
 &\trival{\varphi \rightarrow \psi} = \begin{cases}
                                            \VRAI & \hbox{if } \trival{\varphi} \le \trival{\psi} \\   
                                           0 & \textrm{ otherwise. }   
                                        \end{cases}&
\end{eqnarray*}
\noindent For a theory $\Gamma$, we define $\trival{\Gamma} = \min\lbrace\trival{\varphi}\mid \varphi \in \Gamma \rbrace$. 
We say that an interpretation $\bm{m}$ is a \emph{model} of a formula $\varphi$ if $\trival{\varphi} = \VRAI$.

When a QBF is given in prenex normal form, semantics can be given in terms of QBF-policies. A \emph{QBF-policy} (or policy in \cite{cosfarlanlebmar06}) refers to a strategy or decision rule that dictates how to assign truth values to the variables in a QBF to achieve a desired outcome. Since QBF extends classical logic with second order quantification over the propositional variables, a policy can represent a systematic approach for deciding the truth values of variables depending on the quantifier structure.

In the context of Boolean logic, evaluating a propositional variable $p$ to $\VRAI$ (resp. $\FAUX$) on a formula $\varphi$ is equivalent to replacing all occurrences of $p$ by $\top$ (resp. by $\bot$). Unfortunately, the same approach cannot be represented within the context of many-valued logics~\cite{many-valued}, where there are truth values that are not captured with elements of the language. Our definition of QBF-policies can be adapted easily to the case of many-valued logics, as done in Section~\ref{sec:policies}.

\begin{definition}[QBF-policy \cite{cosfarlanlebmar06}]
The set $TP_{QBF}(\overline{Q})$ of QBF-policies for a binder $\overline{Q}$ is defined, recursively, as follows:
\begin{eqnarray*}
TP_{QBF}(\EMPTYBINDER)&=&\{\lambda\} \\
TP_{QBF}(\exists x_i\ldots Q_n x_n) &=&  \lbrace \eQBFpol{x_i}{v}{\pi_{i+1}} \mid    v \in \lbrace \FAUX, \VRAI \rbrace, \pi_{i+1} \in TP_{QBF}(Q_{i+1} x_{i+1}\ldots Q_n x_n)\rbrace\\
TP_{QBF}(\forall x_i\ldots Q_n x_n) &=& \lbrace \uQBFpol{x_i}{\pi^0_{i+1}}{\pi^1_{i+1}}  \mid  \\
    && \hspace{68pt}\pi^0_{i+1}, \pi^1_{i+1} \in TP_{QBF}(Q_{i+1} x_{i+1}\ldots Q_n x_n)\rbrace
\end{eqnarray*}

\noindent where $\lambda$ represents the \emph{empty QBF-policy}.
The operator ``;'' represents the sequential composition of QBF-policies and $\lbrace \FAUX , \VRAI \rbrace \to \pi_{i+1}$ represents all possible functions from assignments of  $x_i$ to values of $\lbrace \FAUX,\VRAI\rbrace$ to QBF-policies in $TP_{QBF}(Q_{i+1} x_{i+1}\ldots Q_n x_n)$. 
\end{definition}

The satisfaction of a QBF is based on the concept of configuration. A \emph{configuration} is the structure $\tuple{\bm{m}, \pi}$, where $\bm{m}$ is a partial classical interpretation used to store the different assignments obtained during the analysis of the QBF-policy $\pi$. Given a QBF  $Q_1 x_1, \cdots, Q_n x_n \varphi$ and a configuration $\tuple{\bm{m}, \pi}$, the satisfaction relation is defined, by cases, as 

\begin{itemize}
    \item $\tuple{\bm{m},\lambda}\models \varphi$ if $\trival{\varphi}=\VRAI$ 
    \item $\tuple{\bm{m},\pi}\models \exists x Q_{i+1} x_{i+1}\cdots Q_n\; x_n \varphi$ if for some $v \in \lbrace \FAUX, \VRAI\rbrace$, $\pi = \eQBFpol{x}{v}{\pi'}$ and $\tuple{\assign{x}{v},\pi'} \models Q_{i+1} x_{i+1}\cdots Q_n\; x_n \varphi$.   
    \item $\tuple{\bm{m},\pi}\models \forall x Q_{i+1} x_{i+1}\cdots Q_n\; x_n \varphi$ if $\pi = \uQBFpol{x}{\pi^0}{\pi^1}$ and both $\tuple{\assign{x}{\FAUX}, \pi^0} \models Q_{i+1} x_{i+1}\cdots Q_n\; x_n \varphi$ and $\tuple{\assign{x}{\VRAI}, \pi^1} \models Q_{i+1} x_{i+1}\cdots Q_n\; x_n \varphi$.   
\end{itemize}

\begin{definition} 
A QBF-policy $\pi$ is said to satisfy a QBF  $Q_1 x_1, \cdots, Q_n x_n \varphi$ if $\tuple{\EMPTY,\pi} \models Q_1 x_1, \cdots, Q_n x_n \varphi$.
\end{definition}

\begin{example}\label{example:qbf} Let us consider the QBF $\varphi = \forall x\;  \exists z \left((z \to x) \wedge \left( z \vee \neg z\right)\right)$. 
There exist two QBF-policies that satisfy $\varphi$:
$\pi_{(a)}=\uQBFpol{x}{\eGpol{z}{\FAUX}{\lambda}}{\eQBFpol{z}{\FAUX}{\lambda}}$ and
$\pi_{(b)}=\uQBFpol{x}{\eQBFpol{z}{\FAUX}{\lambda}}{\eQBFpol{z}{\VRAI}{\lambda}}$
with a tree-shape representation shown in Figure~\ref{fig:policies:qbf}:

\begin{figure}[h!]\centering
\begin{subfigure}{.35\textwidth}\centering
\begin{tikzpicture}[level distance=1cm]
\node{$x$}
child{ node{$z$} 
    child{ node{$\lambda$}
        edge from parent node[left] {$\FAUX$}
    }
    edge from parent node[left] {$\FAUX$}
}
child{ node{$z$} 
       child{ node{$\lambda$}
            edge from parent node[right] {$\FAUX$}          
       }
       edge from parent node[right] {$1$}
}

;
\end{tikzpicture}
\caption{}
\label{fig:policies:qbfa}
\end{subfigure}
\begin{subfigure}{.35\textwidth}\centering
\begin{tikzpicture}[level distance=1cm]
\node{$x$}
child{ node{$z$} 
    child{ node{$\lambda$}
        edge from parent node[left] {$\FAUX$}
    }
    edge from parent node[left] {$\FAUX$}
}
child{ node{$z$} 
       child{ node{$\lambda$}
            edge from parent node[right] {$\VRAI$}          
       }
       edge from parent node[right] {$\VRAI$}
};
\end{tikzpicture}
\caption{}
\label{fig:policies:qbfb}
\end{subfigure}
 \caption{Two policies for the QBF $\psi = \forall x \exists z \left( \left(z \to x\right) \wedge \left( z \vee \neg z\right)\right)$.}
 \label{fig:policies:qbf}
\end{figure}
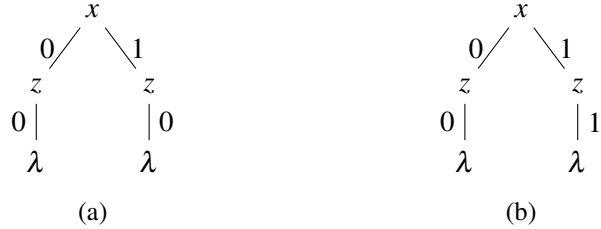
\end{example}
Note that, since $z \vee \neg z$ is a tautology, $\varphi$ is equivalent to $\varphi = \forall x\;  \exists z (z \to x)$. Therefore, it is easy to see that, if $x$ is evaluated to $\FAUX$, we are forced to evaluate $z$ to $\FAUX$. On the contrary, when $x$ is evaluated to $\VRAI$, the quantified formula is satisfied regardless the truth value of $z$. 

\subsection{Propositional Equilibrium Logic}
\label{sec:eqlogic}
EL shares its syntax with propositional logic. Therefore, formulas are generated by means of the grammar presented in~\eqref{grammar:propositional}.
The semantics of EL are defined in terms of the (monotonic) logic of \emph{Here-and-There}~\cite{heyting30a,goedel32a} (\HT{}) together with a minimality criterion that induces the non-monotonicity. However, for our purposes we present an equivalent definition in terms of the three-valued G\"odel logic \GTROIS~\cite{goedel32a}. Within this logic, an interpretation is defined as a three-valued mapping $\bm{m}: \prop \mapsto  \lbrace \FAUX, \UNDEMI, \VRAI\rbrace$. As in classical logic, given $p \in \prop$, $\trival{p}=\FAUX$ means that $p$ is false; $\trival{p}=\VRAI$ means that $p$ is true. 
In addition, $\trival{p}=\UNDEMI$ means that $p$ is true  by default.
Given an  interpretation $\bm{m}$, $x \in \prop$ and $i\in \lbrace \FAUX,\UNDEMI,\VRAI\rbrace$, the concept of \emph{update} is defined now as in~\eqref{eq:update} with the peculiarity that now $i \in \lbrace \FAUX, \UNDEMI,\VRAI\rbrace$. 
The semantics correspond to the propositional case except for the implication $\varphi\to \psi$, which is interpreted as 
\begin{itemize} 
\item $\trival{\varphi \rightarrow \psi} = \begin{cases}
                                            \VRAI & \hbox{if } \trival{\varphi} \le \trival{\psi} \\   
                                            \trival{\psi} & \textrm{ otherwise. }   
                                            \end{cases}$ 
\end{itemize}
Note that, in \GTROIS{} (so in classical logic as well), $\neg \varphi \eqdef \varphi \to \bot$. As in the classical case, we say that a \GTROIS{}  interpretation $\bm{m}$ is a \emph{model} of a formula $\varphi$ if $\trival{\varphi} = \VRAI$.

\begin{proposition}[\cite{pearce06a}] Here-and-There logic and the G\"odel \GTROIS\ are equivalent. 
\end{proposition}

Given an  interpretation $\bm{m}$ and $\Sigma\subseteq \prop$, we define the  interpretation $\crisp{\Sigma}{\bm{m}}$ \footnote{The term crisp has been coined in the area of fuzzy logics to denote ``classical interpretations''} as $\crisp{\Sigma}{\bm{m}}(p)=1$ if $\trival{p}=\UNDEMI$ and $p \in \Sigma$; $\crisp{\Sigma}{\bm{m}}(p)=\trival{p}$, otherwise. 
For simplicity, we will denote $\crisp{~}{\bm{m}}\eqdef \crisp{\prop}{\bm{m}}$.

\begin{proposition}\label{prop:classical} For all  interpretations $\bm{m}$ and  formulas $\varphi$, $\crisp{~}{\bm{m}}(\varphi)\in \lbrace 0,1\rbrace$.
\end{proposition}

\begin{proposition}[persistency]\label{proof:presistency} For all formulas $\varphi$ and for all  interpretations $\bm{m}$,
\begin{enumerate*}[label=(\arabic*)]
\item $\trival{\varphi} \not= \FAUX$  implies $\crisp{~}{\bm{m}}(\varphi)= \VRAI$ and 
\item $\trival{\varphi} = \FAUX$  implies $\crisp{~}{\bm{m}}(\varphi)= \FAUX$.
\end{enumerate*}
\end{proposition}

\begin{corollary}\label{cor:negation} For all formulas $\varphi$ and interpretations $\bm{m}$, $\trival{\neg \varphi}=1$ iff $\crisp{~}{\bm{m}}(\varphi) =0$.
\end{corollary}

Proposition~\ref{prop:classical} states that crisp  interpretations can be regarded as classical interpretations. We can relate two  interpretations $\bm{m}$ and $\bm{m'}$, with respect to $\Sigma\subseteq\prop$, by saying that
\begin{itemize}[itemsep=0pt]
\item $\bm{m} \triangleq_\Sigma \bm{m'}$ if $\trival{p} = \bm{m'}(p)$, for all $p \in \Sigma$. 
\item $\bm{m} \trianglelefteq_\Sigma \bm{m'}$ if $\crisp{\Sigma}{\bm{m}} \triangleq_\Sigma \crisp{\Sigma}{\bm{m'}}$ and $\trival{p} \le \bm{m'}(p)$, for all $p \in \Sigma$.
\item $\bm{m} \triangleleft_{\Sigma} \bm{m'}$ if  $\bm{m} \trianglelefteq_{\Sigma} \bm{m'}$ and not $\bm{m} \triangleq_\Sigma \bm{m'}$.
\end{itemize}

When $\Sigma=\prop$, the suffix $\Sigma$ will be omitted from the context in the relations above. 

\begin{definition}[Equilibrium Model] 
\label{def:EM}
We say that an  interpretation $\bm{m}$ is an equilibrium model of a propositional formula $\varphi$ if the following conditions hold: 
\begin{enumerate}[label=(\arabic*)]
    \item $\TRIVAL = \crisp{\prop}{\TRIVAL}$, i.e. $\bm{m}$ is a classical model.
    \item \TRIVAL\ is $\trianglelefteq_{\prop}$-minimal, i.e. there is no  $\TRIVAL' \triangleleft_{\prop} \TRIVAL$ satisfying $\TRIVAL'(\varphi)=\VRAI$.
\end{enumerate} 
\end{definition}

\begin{example}[Example \ref{example:qbf} continued]\label{example:qel}
Let us consider the formula $\varphi = ((z  \rightarrow x) \wedge ( z \vee \neg z))$. This program corresponds to the logic program \verb|{ x :- z. {z}. }|, which has two answer sets: $\emptyset$ and $\lbrace x, z\rbrace$. The three-valued truth table, containing all possible \GTROIS-interpretations for $\varphi$ is displayed below.  

\[\begin{array}{|c|ccccccccc|}
\hline
i & \textbf{1} & \textbf{2} & \textbf{3} & \textbf{4} & \textbf{5} & \textbf{6}& \textbf{7} & \textbf{8} & \textbf{9} \\\hline
x & \FAUX & \FAUX & \FAUX &\UNDEMI & \UNDEMI & \UNDEMI &  1 & 1 & 1\\
z & \FAUX & \UNDEMI & 1 & \FAUX & \UNDEMI & 1      &  \FAUX & \UNDEMI & 1\\
\hline
\bm{m_i}(\varphi) & 1 & \FAUX & \FAUX & 1 & \UNDEMI & \UNDEMI & 1 & \UNDEMI & 1 \\
\hline
\end{array}\]
\noindent It can be checked that $\bm{m_1}$, $\bm{m_4}$, $\bm{m_7}$ and $\bm{m_9}$ are the models of $\varphi$.
Moreover, $\bm{m_1}$, $\bm{m_7}$ and $\bm{m_9}$ are crisp so they are the candidates to become equilibrium models. 
Both $\bm{m_1}$ and $\bm{m_9}$ are $\trianglelefteq$-minimal so they are equilibrium models that correspond to the answer sets $\emptyset$ and $\lbrace x, z\rbrace$, respectively. 
On the contrary, it can be checked that $\bm{m_4} \triangleleft \bm{m_7}$, so $\bm{m_7}$ is not an equilibrium model. The following propositions show that formulas of the form $\neg x$ (resp. $\neg \neg x$), with $x \in \prop$, can be used to select equilibrium models where $x$ is false (resp. true). 
\end{example}

\begin{proposition}\label{prop:x:positive} Let $\Gamma$ be a propositional theory and let $x$ be a propositional variable. For all \GTROIS-interpretations $\bm{m}$, $\bm{m}$ is an equilibrium model of $\Gamma \cup \lbrace \neg \neg x \rbrace$ iff $\bm{m}$ is an equilibrium model of $\Gamma$ and $\trival{x} = \VRAI$.
\end{proposition}

\begin{proposition}\label{prop:x:negative} Let $\Gamma$ be a propositional theory and let $x$ be a propositional variable. For all \GTROIS- interpretations $\bm{m}$, $\bm{m}$ is an equilibrium model of $\Gamma \cup \lbrace \neg x \rbrace$ iff $\bm{m}$ is an equilibrium model of $\Gamma$ and $\trival{x}=\FAUX$.
\end{proposition}

Finally, it is worth to remark that EL and answer sets semantics coincide.
\begin{theorem}[\cite{pearce06a}] There is a bijection between equilibrium models and answer sets of any arbitrary propositional theory.
\end{theorem}

\subsection{Quantified Answer Set Programming}\label{sec:qasp}

In this section we present the semantics published by~\cite{ua8046} and in~\cite{falaroscso21a}. However, as a difference, we allow any propositional formula (or theory) as a matrix of the QBF. The semantics defined in~\cite{falaroscso21a} are oriented towards satisfiability within answer set semantics and they are defined next.

\begin{definition}[from \cite{falaroscso21a}]\label{def:fandi:semantics}
A logic program $P$ is said to be \emph{satisfiable} if it has an answer set. The satisfiability of a QASP program $\overline{Q}\; P$ is defined as follows:
\begin{enumerate}[itemsep=0pt]
    \item $\exists x \; P$ is satisfiable if either $P \cup \lbrace \neg \neg x \rbrace$ or $P \cup \lbrace \neg x \rbrace$  is satisfiable.
    \item $\forall x \; P$ is satisfiable if both $P \cup \lbrace \neg \neg x \rbrace$  and $P \cup \lbrace \neg x \rbrace$ are satisfiable.
    \item $\exists x \overline{Q} \; P$ is satisfiable if either $\overline{Q}\left( P \cup \lbrace \neg \neg x\rbrace\right)$ or $\overline{Q}\left( P \cup \lbrace \neg x\rbrace\right)$ is satisfiable.
     \item $\forall x \overline{Q} \; P$ is satisfiable if both $\overline{Q}\left( P \cup \lbrace \neg \neg x \rbrace\right)$ and $\overline{Q}\left( P \cup \lbrace \neg x \rbrace\right)$ are satisfiable.
\end{enumerate}
\end{definition}

The satisfiability of a quantified variable is tested by adding $\neg \neg x$ (resp. $\neg x$) to $P$ for every quantified variable $x$. 
In view of Propositions~\ref{prop:x:negative} and~\ref{prop:x:positive}, those formulas allow selecting equilibrium models of $P$ where propositional variables occur either positively or negatively. 
Moreover, they prove the following complexity result.

\begin{definition}[from \cite{ua8046}]\label{def:igor:semantics}
 A logic program $P$ is said to be \emph{satisfiable} if it has an answer set. The satisfiability of a QASP program $\overline{Q}\; P$ is defined as follows: \
\begin{enumerate}[itemsep=0pt]
    \item $\exists x \; P$ is satisfiable if either $P \cup \lbrace x \rbrace$ or $P \cup \lbrace \neg x \rbrace$  is satisfiable.
    \item $\forall x \; P$ is satisfiable if both $P \cup \lbrace x \rbrace$  and $P \cup \lbrace \neg x \rbrace$ are satisfiable.
    \item $\exists x \overline{Q} \; P$ is satisfiable if either $\overline{Q}\left( P \cup \lbrace x\rbrace\right)$ or $\overline{Q}\left( P \cup \lbrace \neg x\rbrace\right)$ is satisfiable.
    \item $\forall x \overline{Q} \; P$ is satisfiable if both $\overline{Q}\left( P \cup \lbrace x\rbrace\right)$ and $\overline{Q}\left( P \cup \lbrace \neg x\rbrace\right)$ are satisfiable.
\end{enumerate}
\end{definition}

Stéphan's semantics differ from Fandinno et al. in the way variables are forced to be true within an answer set. 
In~\cite{ua8046} a variable $x$ is forced to become true by adding $x$ directly to the context. 
Therefore, $x$ becomes directly true. 
In~\cite{falaroscso21a} instead, $x$ is forced to become true by adding $\neg \neg x$ to the context. 
In view of Proposition~\ref{prop:x:positive}, any equilibrium model must satisfy $x$. Therefore, truth value of $x$ must be derivable from program itself. 
The following proposition can be used to show that the semantics of Definition~\ref{def:fandi:semantics} are weaker that the ones of Definition~\ref{def:igor:semantics}. 
 
\begin{proposition}\label{prop:weaker} 
For all theory $\Gamma$ and for all $x \in \prop$ and for all interpretations $\bm{m}$, if $\bm{m}$ is an equilibrium model of $\Gamma \cup \lbrace \neg \neg x\rbrace$ then $\bm{m}$ is an equilibrium model of $\Gamma \cup \lbrace x \rbrace$.
\end{proposition}

\begin{corollary}\label{cor:weaker} 
If a quantified theory $\overline{Q}\Gamma$ is satisfiable w.r.t. Definition~\ref{def:fandi:semantics} then it is satisfiable w.r.t. Definition~\ref{def:igor:semantics}.
\end{corollary}
\begin{lemma}\label{lem:different} The semantics proposed in Definition~\ref{def:fandi:semantics} and Definition~\ref{def:igor:semantics} are not equivalent. \end{lemma}

\begin{proof}

Let us consider the following quantified propositional theory:

\begin{equation}
	\forall x \exists y \exists z \underbrace{\left\{
		z \to x, 
		\neg z \to y,
		\neg y \to z
		\right\}}_{\Gamma} \label{ex:difference}
\end{equation}

The answer sets of $\Gamma$ are $\lbrace x, z\rbrace$ and $\lbrace y \rbrace$.  
From the point of view of satisfiability, the following example is satisfiable with respect to both definitions: the reader can check that the policy of Figure~\ref{fig:different:ok} is accepted by the semantics of both Definition~\ref{def:fandi:semantics} and Definition~\ref{def:igor:semantics}. 
However, if we compute all possible policies associated to~\eqref{ex:difference} we will see that the policy of Figure~\ref{fig:different:nok} is accepted by the semantics of Definition~\ref{def:igor:semantics} but not by the semantics of Definition~\ref{def:fandi:semantics}. 
Notice that the policy of Figure~\ref{fig:different:nok} would lead to the evaluation of the answer set $\lbrace x,y\rbrace$, which is not an answer set of $\Gamma$. Since the first quantifier in~\eqref{ex:difference} is universal, the system \texttt{asp2qbf} only answers satisfiable but it does not provide more information allowing us to determine that the policy in Figure~\ref{fig:different:nok} is not generated by the semantics of Definition~\ref{def:fandi:semantics}.

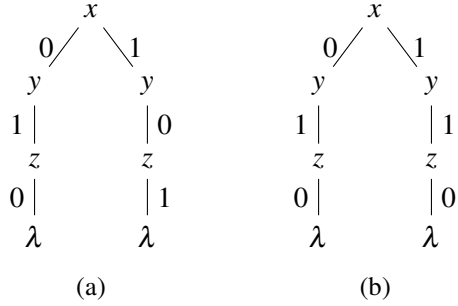
\begin{figure}[h!]\centering
	\begin{subfigure}{.23\textwidth}\centering
		\begin{tikzpicture}[level distance=1cm]
			\node{$x$}
			child{ node{$y$} 
				child{ node{$z$}
					child{ node{$\lambda$}
						edge from parent node[left] {$\FAUX$} 
					}
					edge from parent node[left] {$\VRAI$}
				}
				edge from parent node[left] {$\FAUX$}
			}
			child{ node{$y$} 
				child{ node{$z$}
					child{node{$\lambda$}
						edge from parent node[right] {$\VRAI$}          
					}
					edge from parent node[right] {$\FAUX$}          
				}
				edge from parent node[right] {$\VRAI$}
			};
		\end{tikzpicture}
		\caption{}
		\label{fig:different:ok}
	\end{subfigure} \begin{subfigure}{.23\textwidth}\centering
		\begin{tikzpicture}[level distance=1cm]
			\node{$x$}
			child{ node{$y$} 
				child{ node{$z$}
					child{node{$\lambda$}
						edge from parent node[left] {$\FAUX$}          
					}
					edge from parent node[left] {$\VRAI$}          
				}
				edge from parent node[left] {$\FAUX$}
			}
			child{ node{$y$} 
				child{ node{$z$}
					child{ node{$\lambda$}
						edge from parent node[right] {$\FAUX$} 
					}
					edge from parent node[right] {$\VRAI$}
				}
				edge from parent node[right] {$\VRAI$}
			};
		\end{tikzpicture}
		\caption{}
		\label{fig:different:nok}
	\end{subfigure}
	\caption{Two policies for the quantified theory~\eqref{ex:difference}.}
	\label{fig:policies:g3}
\end{figure}
\end{proof}

 \section{\QGTROIS\ Policies}\label{sec:policies} 
In the first place, we provide an interpretation of the second-order quantifiers within the context of \GTROIS{}: given a \GTROIS{} interpretation $\bm{m}$, the semantics of  quantifiers is defined as 
\begin{displaymath}
\trival{\forall x \varphi} = \min\left\{ \assigns{x}{i}{\varphi} \mid i \in \lbrace\FAUX,\frac{1}{2},\VRAI \rbrace\right\} \hspace{20pt} \trival{\exists x \varphi} = \max\left\{ \assigns{x}{i}{\varphi}\mid i \in \lbrace\FAUX,\frac{1}{2},\VRAI \rbrace\right\}.
\end{displaymath}
\noindent The resulting logic is called \emph{Quantified} \GTROIS{} (\QGTROIS) and follows~\cite{BaazCZ00}.
We extend the concept of QBF-policy to the case of \QGTROIS.
\begin{definition}[\QGTROIS-policy]
The set $TP_{QG3}(\overline{Q})$ of \QGTROIS-policies for a binder $\overline{Q}$ is defined, recursively, as follows:
\begin{eqnarray*}
    TP_{QG3}(\varepsilon)&=&\{\lambda\} \\
    TP_{QG3}(\exists x_i\ldots Q_n x_n) &=&  \lbrace \eGpol{x_i}{v}{\pi_{i+1}} \mid   v \in \lbrace \FAUX, \UNDEMI, \VRAI \rbrace, \pi_{i+1} \in TP_{QG3}(Q_{i+1} x_{i+1}\ldots Q_n x_n)\rbrace\\
    TP_{QG3}(\forall x_i\ldots Q_n x_n) & = & \lbrace \uGpol{x_i}{\pi^0_{i+1}}{\pi^{\UNDEMI}_{i+1}}{\pi^1_{i+1}}  \mid  \\
    && \hspace{50pt}\pi^0_{i+1}, \pi^{\UNDEMI}_{i+1}, \pi^1_{i+1} \in TP_{QG3}(Q_{i+1} x_{i+1}\ldots Q_n x_n)\rbrace.
\end{eqnarray*}
The operator ``;'' represents the sequential composition of \QGTROIS-policies and $\lbrace \FAUX , \UNDEMI, 1 \rbrace \to \pi_{i+1}$ refers to all possible functions from assignments of $x_i$ to values of $\lbrace \FAUX,\UNDEMI,1\rbrace$ to \QGTROIS-policies in $TP_{QG3}(Q_{i+1} x_{i+1}\ldots Q_n x_n)$. 
{A \GTROIS\ interpretation $\bm{m}$ is in a \QGTROIS-policy $\pi$ w.r.t. a binder $\overline{Q}$ ($\TRIVAL\in_{\overline{Q}} \pi$ in symbols), if:
\begin{enumerate}[label=\arabic*)]
    \item $\bm{m}\in_{\exists x \overline{Q}} \{ \eGpol{x}{v}{\pi} \}$ if $\trival{x}=v$ and $\bm{m}\in_{\overline{Q}} \pi$
    \item $\bm{m}\in_{\forall x \overline{Q}} \uGpol{x}{\pi^0}{\pi^{\UNDEMI}}{\pi^1}$ if $\trival{x}=v$ and $\bm{m}\in_{\overline{Q}} \pi^v$.
\end{enumerate}
By extension, a \GTROIS\ interpretation \TRIVAL\ is in a configuration $\tuple{\TRIVAL',\pi}$ w.r.t. a binder $\overline{Q}$, $\TRIVAL \in_{\overline{Q}} \tuple{\TRIVAL',\pi}$ in symbols, if $\TRIVAL\in_{\overline{Q}} \pi$ and $\TRIVAL \egalite{\prop\setminus \prop(\overline{Q})} \TRIVAL'$.
}

\end{definition}
In \QGTROIS{}, a \emph{configuration} is still a structure $\tuple{\bm{m}, \pi}$, but now $\bm{m}$ is a partial \GTROIS{} interpretation used to store the different assignments obtained during the analysis of the \QGTROIS-policy $\pi$. Given a QBF  $Q_1 x_1, \cdots, Q_n x_n \varphi$ and a configuration $\tuple{\bm{m}, \pi}$, the satisfaction relation is defined as 

\begin{itemize}
    \item $\tuple{\bm{m},\lambda}\models \varphi$ if $\trival{\varphi}=\VRAI$ 
    \item $\tuple{\bm{m},\pi}\models \exists x Q_{i+1} x_{i+1}\cdots Q_n\; x_n \varphi$ if for some $v \in \lbrace \FAUX, \UNDEMI, \VRAI\rbrace$, $\pi = \eQBFpol{x}{v}{\pi'}$  and $\tuple{\assign{x}{v},\pi'} \models Q_{i+1} x_{i+1}\cdots Q_n\; x_n \varphi$
    \item $\tuple{\bm{m},\pi}\models \forall x Q_{i+1} x_{i+1}\cdots Q_n\; x_n \varphi$ if $\pi = \uGpol{x}{\pi^0}{\pi^\UNDEMI}{\pi^1}$ and $\tuple{\assign{x}{\FAUX}, \pi^0} \models Q_{i+1} x_{i+1}\cdots Q_n\; x_n \varphi$, $\tuple{\assign{x}{\UNDEMI}, \pi^\UNDEMI} \models Q_{i+1} x_{i+1}\cdots Q_n\; x_n \varphi$ and $\tuple{\assign{x}{\VRAI}, \pi^1} \models Q_{i+1} x_{i+1}\cdots Q_n\; x_n \varphi$.
\end{itemize}

\begin{definition} 
A \QGTROIS-policy $\pi$ is said to satisfy a QBF  $Q_1 x_1 \cdots Q_n x_n \varphi$ if $\tuple{\EMPTY,\pi} \models Q_1 x_1 \cdots Q_n x_n \varphi$.
\end{definition}
\begin{example}[Example \ref{example:qel} continued]\label{example:qel2}
	Let us consider the formula $\psi = \forall x \exists z \left( \left(z \to x\right) \wedge \left( z \vee \neg z\right)\right)$. 
	The two \QGTROIS-policies satisfying $\psi$ are 
	$\pi_{(a)}=\uGpol{x}{\eGpol{z}{\FAUX}{\lambda}}{\eGpol{z}{\FAUX}{\lambda}}{\eGpol{z}{\FAUX}{\lambda}}$ and
	$\pi_{(b)}=\uGpol{x}{\eGpol{z}{\FAUX}{\lambda}}{\eGpol{z}{\FAUX}{\lambda}}{\eGpol{z}{\VRAI}{\lambda}}$
	and their tree-shape representations are shown in Figure~\ref{fig:policies:qg3}.
	\begin{figure}[h!]\centering
		\begin{subfigure}{.30\textwidth}

			\begin{tikzpicture}[level distance=1cm]
				\node{$x$}
				child{ node{$z$} 
					child{ node{$\lambda$}
						edge from parent node[left] {$\FAUX$}
					}
					edge from parent node[left] {$\FAUX$}
				}
				child{ node{$z$} 
					child{ node{$\lambda$}
						edge from parent node[left] {$\FAUX$}   
					}
					edge from parent node[left] {$\UNDEMI$}
				}
				child{ node{$z$} 
					child{ node{$\lambda$}
						edge from parent node[right] {$\FAUX$}          
					}
					edge from parent node[right] {$1$}
				}
				
				;
			\end{tikzpicture}
			\caption{}
			\label{fig:policies:qg3a}
		\end{subfigure}\begin{subfigure}{.30\textwidth}

			\begin{tikzpicture}[level distance=1cm]
				\node{$x$}
				child{ node{$z$} 
					child{ node{$\lambda$}
						edge from parent node[left] {$\FAUX$}
					}
					edge from parent node[left] {$\FAUX$}
				}
				child{ node{$z$} 
					child{ node{$\lambda$}
						edge from parent node[left] {$\FAUX$}   
					}
					edge from parent node[left] {$\UNDEMI$}
				}
				child{ node{$z$} 
					child{ node{$\lambda$}
						edge from parent node[right] {$1$}          
					}
					edge from parent node[right] {$1$}
				};
			\end{tikzpicture}
			\caption{}
			\label{fig:policies:qg3b}
		\end{subfigure}
		\caption{Two \QGTROIS-policies satisfying the QBF $\psi = \forall x \exists z \left( \left(z \to x\right) \wedge \left( z \vee \neg z\right)\right)$.}
		\label{fig:policies:qg3}
	\end{figure}
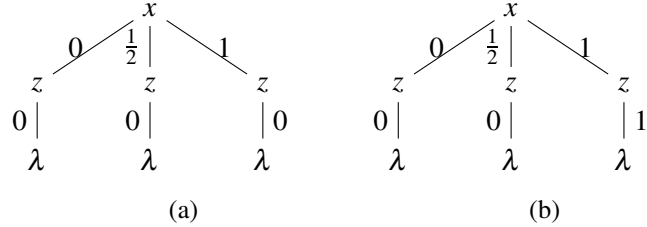
\end{example}

The following lemma shows that satisfiability according to \QGTROIS-semantics of a QBF $\overline{Q} \varphi$ with no free variables (ie a variable that has an occurrence in $\varphi$ but no one in $\overline{Q}$) depends only on the variables of $\varphi$. 

\begin{lemma}
Let $Q_1 x_1 \cdots Q_n x_n \varphi$ a QBF such that $\prop(\varphi)=\{x_1,\dots,x_n\}$ and \TRIVAL\ a \GTROIS\ interpretation.
    $\tuple{\EMPTY,\pi} \models Q_1 x_1 \cdots Q_n x_n \varphi$ if and only if $\tuple{\TRIVAL,\pi} \models Q_1 x_1 \cdots Q_n x_n \varphi$.
\end{lemma}

The notion of minimal (or equilibrium) model for a propositional formula $\varphi$ relies on the partial order $\trianglelefteq$ induced among the different interpretations. In the case of QBFs the truth values of the propositional variables are fixed by the interpretation of the quantifiers. As a consequence, its answer sets are represented in terms of \emph{equilibrium configurations}, defined below.

\begin{definition}[Equilibrium Configuration]
\label{def:ec}
Given a QBF $\overline{Q} \varphi$ and a configuration $\tuple{\bm{m}, \pi}$, we define the notion of  $\tuple{\bm{m}, \pi}$ being an equilibrium configuration recursively as follows:

 \begin{itemize}
    \item If $\overline{Q}\varphi$ has no quantifiers then $\tuple{\bm{m},\pi}$ is an equilibrium configuration of $\overline{Q}\varphi$ if $\pi=\lambda$ and $\bm{m}$ is an equilibrium model of $\varphi$.
    \item If $\overline{Q}\varphi$ is of the form $\exists x Q_{i+1} x_{i+1}\cdots Q_n x_n \varphi$ then $\tuple{\bm{m},\pi}$ is an equilibrium configuration of $\overline{Q}\varphi$ if both 
    \begin{enumerate*}[label=\arabic*)]
        \item $\pi = \eQBFpol{x}{v}{\pi'}$ with $v \in \lbrace \FAUX, 1\rbrace$ and 
        \item $\tuple{\assign{x}{v},\pi'}$ is an equilibrium configuration of $Q_{i+1} x_{i+1}\cdots Q_n x_n \varphi$
    \end{enumerate*}    
    \item If $\overline{Q}\varphi$ is of the form $\forall x Q_{i+1} x_{i+1}\cdots Q_n\; x_n \varphi$ then $\tuple{\bm{m},\pi}$ is an equilibrium configuration of $\overline{Q}\varphi$ if $\pi = \uQBFpol{x}{\pi^0}{\pi^1}$ and both 
        \begin{enumerate*}[label=\arabic*)]
            \item $\tuple{\assign{x}{\FAUX}, \pi^0}$ is an equilibrium configuration of $Q_{i+1} x_{i+1}\cdots Q_n x_n \varphi$ and 
            \item $\tuple{\assign{x}{\VRAI}, \pi^1}$ is an equilibrium configuration of $Q_{i+1} x_{i+1}\cdots Q_n x_n \varphi$. 
        \end{enumerate*}
\end{itemize}
\end{definition}
A QBF-policy $\pi$ is said to be an equilibrium policy of a quantified Boolean formula $\overline{Q} \varphi$ if there exists an equilibrium configuration $\tuple{\bm{m},\pi}$ of $\overline{Q}\varphi$.

\section{Algorithm for Equilibrium Policies}\label{sec:algorithmic}
We propose an algorithm that computes the QBF-policies of a QBF according to the equilibrium policies of the previous section.
Our algorithm uses configurations $\tuple{\TRIVAL,\pi}$ where \TRIVAL\ is a \GTROIS\ interpretation and $\pi$ is a QBF-policy instead of a \GTROIS-policy in order to compute the equilibrium policies.
To do this, we modify the satisfaction of QBFs as follows (the case \UNDEMI\ is deleted from the semantics of the quantifier but not from the interpretation \TRIVAL\ from \GTROIS):

\begin{itemize}
    \item $\tuple{\bm{m},\lambda}\models \varphi$ if $\trival{\varphi}=\VRAI$ 
    \item $\tuple{\bm{m},\pi}\models \exists x Q_{i+1} x_{i+1}\cdots Q_n\; x_n \varphi$ if  
    \begin{enumerate*}[label=\arabic*)]
        \item $\pi = \eQBFpol{x}{v}{\pi'}$, with $v \in \lbrace \FAUX, \VRAI\rbrace$ and 
        \item $\tuple{\assign{x}{v},\pi'} \models Q_{i+1} x_{i+1}\cdots Q_n\; x_n \varphi$
    \end{enumerate*}
    \item $\tuple{\bm{m},\pi}\models \forall x Q_{i+1} x_{i+1}\cdots Q_n\; x_n \varphi$ if $\pi = \uQBFpol{x}{\pi^0}{\pi^1}$ and both   
    \begin{enumerate*}[label=\arabic*)]
        \item $\tuple{\assign{x}{\FAUX}, \pi^0} \models Q_{i+1} x_{i+1}\cdots Q_n\; x_n \varphi$ and  
         \item $\tuple{\assign{x}{\VRAI}, \pi^1} \models Q_{i+1} x_{i+1}\cdots Q_n\; x_n \varphi$.
    \end{enumerate*}
\end{itemize}
Our algorithm realises the minimization thanks to a new structure, named $\Theta$-pair, that captures some (potential) sub-policies plus the non-crisp models that may delete by minimisation some crisp models (and, by the same token, some potential policies):

\begin{definition}[$\Theta^{\varphi}_{\overline{Q}}$-pair conditioned by an interpretation \TRIVAL]
\label{def:theta}

Let $\overline{Q}$ be a binder and $\varphi$ a quantifier-free formula.
    A $\Theta^{\varphi}_{\overline{Q}}$-pair conditioned by the interpretation \TRIVAL\ is a pair $(T,HC)$ consisting of a set of QBF-policies $T\subset TP_{QBF}(\overline{Q})$ such that for all $\pi \in T$, $\tuple{\TRIVAL,\pi}\models \overline{Q}\varphi$,  and a set $HC$ of non-crisp models of $\varphi$ such that for all $\TRIVAL' \in HC, \TRIVAL' \egalite{\prop(\varphi)\setminus \prop(\overline{Q})} \TRIVAL$ and  $\TRIVAL' \not\triangleleft_{\prop(\overline{Q})} \TRIVAL$.
    
\end{definition}

The following example illustrates the base cases $\overline{Q} = \exists z$ and $\overline{Q}=\forall z$ of our quantifier elimination algorithm for QASP.

\begin{example}[Example \ref{example:qel} continued]\label{example:HC} 
Let $\varphi = \left( \left(z \to x\right) \wedge \left( z \vee \neg z\right)\right)$ be the formula of Example~\ref{example:HC}. 
We consider a binder of the form $Q x\; Q z$ meaning that the variable $x$ is eliminated before $z$. In that case, we assume that $x$ has been (already) assigned with a value during the elimination process. If $Q z = \exists z$ we will reason as follows:
\begin{itemize}[itemsep=0pt]
    \item Since $\tuple{\TRIVAL_1,\lambda}\models\varphi$ and $\TRIVAL_2(\varphi)=\TRIVAL_3(\varphi)=\FAUX$  then $\tuple{\EMPTY^x_{\FAUX},\eQBFpol{z}{\FAUX}{\lambda}}\models \exists z\varphi$ and $(T^{\exists z}_{\FAUX}, HC^{\exists z}_{\FAUX}) = (\{\eQBFpol{z}{\FAUX}{\lambda}\},\emptyset)$ is a $\Theta^{\varphi}_{\exists z}$-pair conditioned by $\EMPTY^x_{\FAUX}$;
    \item Since $\TRIVAL_5(\varphi)=\TRIVAL_6(\varphi)=\UNDEMI$  and $\TRIVAL_4(\varphi)=\VRAI$ and $\TRIVAL_4$ is a non-crisp model of $\varphi$ then  $(T^{\exists z}_{\UNDEMI},HC^{\exists z}_{\UNDEMI})=(\emptyset,\{\TRIVAL_4\})$ is a $\Theta^{\varphi}_{\exists z}$-pair conditioned by $\EMPTY^x_{\UNDEMI}$; 
    \item Since $\tuple{\TRIVAL_7,\lambda}\models\varphi$  and  $\tuple{\TRIVAL_9,\lambda}\models\varphi$ and $\TRIVAL_8(\varphi)=\UNDEMI$ then $\tuple{\EMPTY^x_{\VRAI},\eQBFpol{z}{\FAUX}{\lambda}} \models \exists z\varphi$ and $\tuple{\EMPTY^x_{\VRAI},\eQBFpol{z}{\VRAI}{\lambda}} \models \exists z\varphi$ and $(T^{\exists z}_{\VRAI}, HC^{\exists z}_{\VRAI})=(\{\eQBFpol{z}{\FAUX}{\lambda},\eQBFpol{z}{\VRAI}{\lambda}\},\emptyset)$ is a $\Theta^{\varphi}_{\exists z}$-pair conditioned by $\EMPTY^x_{\VRAI}$.
\end{itemize}
If  $Q z = \forall z$ instead, we get
\begin{itemize}[itemsep=0pt]
\item Since $\TRIVAL_3(\varphi)=\FAUX$ and $\TRIVAL_2(\varphi)=\FAUX$ then $(T^{\forall z}_{\FAUX}, HC^{\forall z}_{\FAUX})=(\emptyset,\emptyset)$ is a $\Theta^{\varphi}_{\forall z}$-pair conditioned by $\EMPTY^x_{\FAUX}$;
    \item Since $\TRIVAL_4(\varphi)=\VRAI$  and $\TRIVAL_4$ is a non-crisp interpration  then $(T^{\forall z}_{\UNDEMI},HC^{\forall z}_{\UNDEMI})=(\emptyset,\{\TRIVAL_4\})$ is a $\Theta^{\varphi}_{\forall z}$-pair conditioned by $\EMPTY^x_{\UNDEMI}$; 
    \item Since $\tuple{\TRIVAL_7,\lambda}\models\varphi$  and  $\tuple{\TRIVAL_9,\lambda}\models\varphi$ and $\TRIVAL_8(\varphi)=\UNDEMI$ then $\tuple{\EMPTY^x_{\VRAI},\uQBFpol{x}{\lambda}{\lambda}} \models \forall z\varphi$ and $(T^{\forall z}_{\VRAI},HC^{\forall z}_{\VRAI})=(\{\uQBFpol{x}{\lambda}{\lambda}\},\emptyset)$ is a $\Theta^{\varphi}_{\forall z}$-pair conditioned by $\EMPTY^x_{\VRAI}$.
\end{itemize}
\end{example}

Our algorithm, named \QEM and reported in Algorithm \ref{alg:QEM}, is a quantifier elimination algorithm from the outermost to the innermost quantifier. 
At each call, the \QEM\ algorithm constructs as output a $\Theta^{\varphi}_{\overline{Q}}$-pair conditioned by \TRIVAL.
The output of the initial call $\qem{\overline{Q}}{\varphi}{\EMPTY}$ is a $\Theta^{\varphi}_{\overline{Q}}$-pair conditioned by \EMPTY, $(T,HC)$, such that $T$ is the set of the equilibrium policies of $\overline{Q}\varphi$.

Base cases for the existential quantifier (lines 4 to 6) and the universal quantifier (lines 9 to 13) of the \QEM\ algorithm apply Definition~\ref{def:ec} to binders with only one quantifier following a restrictive application of minimality to the quantified variable:  
\begin{itemize}[itemsep=0pt]
    \item If $\tuple{\bm{m},\pi}$ is an equilibrium configuration of $\exists x \varphi$ then $\pi = \eQBFpol{x}{v}{\lambda}$ with  
    \begin{enumerate}
\item $v=\FAUX$ and $\assign{x}{v}(\varphi)=\VRAI$ or
        \item $v=\VRAI$ and $\assign{x}{v}(\varphi)=\VRAI$ and  $\assign{x}{\UNDEMI}(\varphi) \leq \VRAI$.
    \end{enumerate}    
    \item If $\tuple{\bm{m},\pi}$ is an equilibrium configuration of $\forall x \varphi$ then $\pi = \uQBFpol{x}{\lambda}{\lambda}$ with
        \begin{enumerate}
\item $\assign{x}{\FAUX}(\varphi) = \VRAI$ and 
            \item $\assign{x}{\VRAI}(\varphi) = \VRAI$  and  $\assign{x}{\UNDEMI}(\varphi) \leq \VRAI$. 
        \end{enumerate}

\end{itemize}
Equilibrium configurations cannot be constructed exclusively based on the interpretation \TRIVAL.
Clearly, the first item of Definition~\ref{def:ec} imposes a more global condition: given a quantified formula $\overline{Q}\varphi$ and an equilibrium configuration $\tuple{\TRIVAL,\pi}$ of $\overline{Q}\varphi$, every $\TRIVAL'\in_{\overline{Q}} \tuple{\TRIVAL,\pi}$ is an equilibrium model of $\varphi$. 
This condition is fulfilled in our algorithm thanks to the non-crisp part of the $\Theta^{\varphi}_{\overline{Q}}$-pair conditioned by the updates of the interpretation \TRIVAL.

\begin{example}[Example \ref{example:HC} continued]\label{example:noncrisp}
$\tuple{\EMPTY^x_{\VRAI},\eQBFpol{z}{\FAUX}{\lambda}} \models \exists z\varphi$ but $\TRIVAL_7\in_{\exists z}\tuple{\EMPTY^x_{\VRAI},\eQBFpol{z}{\FAUX}{\lambda}}$ and $\TRIVAL_4 \triangleleft \TRIVAL_7$, $\TRIVAL_4$ is a non-crisp model of $\varphi$, element of the non-crisp part of $(\{\eQBFpol{z}{\FAUX}{\lambda}\},\{\TRIVAL_4\})$, the $\Theta^{\varphi}_{\exists z}$-pair conditioned by $\EMPTY^x_{\UNDEMI}$. 
So $\bm{m_7}$ is not an equilibrium model and $\tuple{\EMPTY,\eQBFpol{x}{\VRAI}{\eQBFpol{z}{\FAUX}{\lambda}}} \not\models \exists x \exists z\varphi$.
\end{example}
We define two new operators, the \MCEXISTS\ operator for $\exists$ quantifier and the \MCFORALL\ operator for $\forall$ quantifier, to combine, after the elimination of the quantifier $Qx$, the $\Theta^{\varphi}_{\overline{Q}}$-pairs $(T_0,HC_0)$, $(T_{\UNDEMI},HC_{\UNDEMI})$ and $(T_1,HC_1)$ conditioned by resp. $\assign{x}{\FAUX}$, $\assign{x}{\UNDEMI}$ and $\assign{x}{\VRAI}$, in a new $\Theta^{\varphi}_{Qx\overline{Q}}$-pairs conditioned by \TRIVAL.

\begin{definition}[\MCEXISTS\ operator]
\label{def:mcexists}
Let $\overline{Q}$ be a binder and $\varphi$ a quantifier-free formula and $x \in \prop(\varphi) \setminus \prop(\overline{Q})$.
Let $(T_0,HC_0)$, $(T_{\UNDEMI},HC_{\UNDEMI})$, and $(T_1,HC_1)$ some $\Theta_{\overline{Q}}$-pairs conditioned by the interpretations,  respectively, \assign{x}{\FAUX}, \assign{x}{\UNDEMI} and \assign{x}{\VRAI}.
The \MCEXISTS\ operator is defined as follows:
$\mcexists{x}{(T_0,HC_0)}{(T_{\UNDEMI},HC_{\UNDEMI})}{\linebreak[1](T_1,HC_1)} = ({T}_{\exists x},{HC}_{\exists x})$ where
\begin{eqnarray}
{T}_{\exists x} & = & \lbrace \eQBFpol{x}{\FAUX}{\pi} ~|~ \pi \in T_0\} \nonumber\\
                      & \cup  & \lbrace \eQBFpol{x}{\VRAI}{\pi} \mid \pi \in T_1,\; \TRIVAL_1 \in_{\BINDER} \pi,\; \TRIVAL_1 \egalite{\prop(\varphi)\setminus\prop(\exists x\overline{Q})} \TRIVAL,\nonumber \\
                      && \hspace{45pt} \hbox{ there is no } \TRIVAL' \in HC_{\UNDEMI} \hbox{ s.t. } \TRIVAL' \triangleleft \TRIVAL_1 \hbox{ and } \label{mcE:hcundemi}\\
                      &&  \hspace{45pt} \hbox{ there is no } \pi' \in T_{\UNDEMI} \hbox{ s.t. }  \TRIVAL_1 \in_{\BINDER} \pi'\rbrace \label{mcE:tundemi}\\
{HC}_{\exists x} &=&      
     \lbrace \TRIVAL' \mid \pi \in T_{\UNDEMI} \hbox{ and } \TRIVAL'\in_{\BINDER} \pi\rbrace \cup \bigcup \limits_{i\in \{\FAUX, \UNDEMI, \VRAI\}} HC_{i}\nonumber
\end{eqnarray}
\end{definition}
By Definition~\ref{def:theta} and the definition of $\triangleleft$, if $\pi \in T_{\FAUX}$ then $\FAUX;\pi \in T_{\exists x}$.
If there exists a $\TRIVAL'$ such that~\eqref{mcE:hcundemi} holds or $\pi'$ is such that~\eqref{mcE:tundemi} holds, then \TRIVAL\ is no more $\trianglelefteq$-minimal and $\tuple{\TRIVAL, \VRAI; \pi}$ cannot be an equilibrium configuration, otherwise $\VRAI;\pi \in T_{\exists x}$.

\begin{example}[Example \ref{example:HC} continued]\label{example:mcE}
Given the QBF $\exists x \exists z \varphi$, the $\MCEXISTS$  applied on $x$ returns  
\begin{equation*}
 \mcexists{x}{(T^{\exists z}_0,HC^{\exists z}_0)}{(T^{\exists z}_{\UNDEMI},HC^{\exists z}_{\UNDEMI})}{(T^{\exists z}_1,HC^{\exists z}_1)} =  (\{\eQBFpol{x}{\FAUX}{\eQBFpol{z}{\FAUX}{\lambda}}, \eQBFpol{x}{\VRAI}{\eQBFpol{z}{\VRAI}{\lambda}}\},\{\TRIVAL_4\}).  
\end{equation*}
\end{example}

The \MCFORALL\ operator is similar to the \MCEXISTS\ operator. 
The difference is that the \MCFORALL\ operator have simultaneously the both conditions on $T_0$ and $T_1$.

\begin{definition}[\MCFORALL\ operator]
\label{def:mcforall}
Let $\overline{Q}$ be a binder and $\varphi$ a quantifier-free formula and $x \in \prop(\varphi) \setminus \prop(\overline{Q})$.
Let $(T_0,HC_0)$, $(T_{\UNDEMI},HC_{\UNDEMI})$, and $(T_1,HC_1)$ some $\Theta_{\overline{Q}}$-pairs conditioned by the interpretations respectively, \assign{x}{\FAUX}, \assign{x}{\UNDEMI} and \assign{x}{\VRAI}.
The \MCFORALL\ operator is defined as follows:
$\mcforall{x}{(T_0,HC_0)}{(T_{\UNDEMI},HC_{\UNDEMI})}{\linebreak[1](T_1,HC_1)} = ({T}_{\forall x},{HC}_{\forall x})$ where
\begin{eqnarray}
{T}_{\forall x} &=& \lbrace \uQBFpol{x}{\pi_0}{\pi_1} \mid \pi_0 \in T_0,\; \pi_1 \in T_1, \nonumber\\
    && \hspace{45pt} \TRIVAL_1 \in_{\BINDER} \pi_1,\; \TRIVAL_1 \egalite{\prop(\varphi)\setminus\prop(\forall x\overline{Q})} \TRIVAL, \nonumber\\
    &&\hspace{45pt} \hbox{there is no } \TRIVAL' \in HC_{\UNDEMI} \hbox{ s.t. } \TRIVAL' \triangleleft \TRIVAL_1 \hbox{ and }\nonumber\\
    && \hspace{45pt} \hbox{there is no } \pi' \in T_{\UNDEMI} \hbox{ s.t. } \TRIVAL_1 \in_{\BINDER}\pi'\nonumber \rbrace
\end{eqnarray}
\noindent and ${HC}_{\forall x}$ is defined as ${HC}_{\exists x}$.
\end{definition}

\begin{example}[Example \ref{example:HC} continued]\label{example:mcV}
Given  $\varphi=\left( \left(z \to x\right) \wedge \left( z \vee \neg z\right)\right)$, we conclude that, for the formula $(\forall x \exists z \varphi)$,
\begin{equation*}
 \mcforall{x}{(T^{\exists z}_0,HC^{\exists z}_0)}{(T^{\exists z}_{\UNDEMI},HC^{\exists z}_{\UNDEMI})}{(T^{\exists z}_1,HC^{\exists z}_1)} =    (\{\uQBFpol{x}{\eQBFpol{z}{\FAUX}{\lambda}}{\eQBFpol{z}{\VRAI}{\lambda}}\},\{\TRIVAL_4\}).
 \end{equation*}
\end{example}

\begin{algorithm}[h!]
\SetKwInOut{Input}{input}\SetKwInOut{Output}{output}\SetKwData{T}{$T_x$}\SetKwData{HC}{$HC_x$}
\Input{A non empty binder $\overline{Q}$, a propositional formula $\varphi$ and a $\GTROIS$ interpretation $\TRIVAL$.}
\Output{A $\Theta^{\varphi}_{\overline{Q}}$-pair conditioned by \TRIVAL}
$\T \gets \emptyset$\;
$\HC \gets \emptyset$\;
\uIf{$\overline{Q}= \exists x$}{
\lIf{$\inters{x}{\UNDEMI}{\varphi} = \VRAI$}{$\HC \gets  \HC \cup \{\assign{x}{\UNDEMI}\}$} 
\lIf{$\inters{x}{\FAUX}{\varphi} = \VRAI$}{$\T \gets  \T \cup \{\eQBFpol{x}{\FAUX}{\lambda}\}$}
\lIf{$\inters{x}{\UNDEMI}{\varphi} \leq \UNDEMI$ and $\inters{x}{1}{\varphi} = \VRAI$}{$\T \gets  \T \cup \{\eQBFpol{x}{\VRAI}{\lambda}\}$}
\Return{$(\T,\HC)$}\;
}
\uIf{$\overline{Q}= \forall x$}{

\lIf{$\inters{x}{\FAUX}{\varphi} = \VRAI$ and $\inters{x}{\UNDEMI}{\varphi} \leq \UNDEMI$ and $\inters{x}{1}{\varphi} =\VRAI$}{
    \Return{$(\{\uQBFpol{x}{\lambda}{\lambda}\},\emptyset)$}
}    
\lIf{$\inters{x}{\UNDEMI}{\varphi} = \VRAI$}{$\HC \gets \HC \cup \{\assign{x}{\UNDEMI}\}$}  
\lIf{$\inters{x}{\FAUX}{\varphi} = \VRAI$ and $\assign{x}{\FAUX}$ is not crisp}{$\HC \gets \HC \cup \{\assign{x}{\FAUX}\}$}  
\lIf{$\inters{x}{\VRAI}{\varphi} = \VRAI$ and $\assign{x}{\VRAI}$ is not crisp}{$\HC \gets \HC \cup \{\assign{x}{\VRAI}\}$}  
\Return{$(\emptyset, \HC)$}\;
}
\uIf{$\overline{Q} = Qx \overline{Q'}$}{
    $(T_0, HC_0) = QEM(\overline{Q'},\varphi,\bm{m}^{x}_{\FAUX})$\;
    $(T_{\UNDEMI}, HC_{\UNDEMI}) = QEM(\overline{Q'},\varphi,\bm{m}^{x}_{\UNDEMI})$\;
    $(T_1, HC_1) = QEM(\overline{Q'},\varphi,\bm{m}^{x}_{1})$\;
    \lIf{$Q=\exists$}{
     \Return{$\mcexists{x}{(T_0, HC_0)}{(T_{\UNDEMI}, HC_{\UNDEMI})}{(T_1, HC_1)} $}
    }
    \lElse{ \Return{$\mcforall{x}{(T_0, HC_0)}{(T_{\UNDEMI},HC_{\UNDEMI})}{(T_1, HC_1)} $}}
}
\caption{$QEM$}\label{alg:QEM}
\end{algorithm}

\begin{theorem}\label{th:QEM}
Given a QBF $\overline{Q} \varphi$ and a \GTROIS{} interpretation \TRIVAL, 
\qem{\overline{Q}}{\varphi}{\TRIVAL} computes the set of equilibrium policies of $\overline{Q} \varphi$ conditioned by $\TRIVAL$.
\end{theorem}

\begin{example}[Examples \ref{example:mcE} and \ref{example:mcV} continued]\label{example:policies}
With  $\varphi=\left( \left(z \to x\right) \wedge \left( z \vee \neg z\right)\right)$ we have that
\begin{enumerate}
\item $\exists x \exists z \varphi$ has two equilibrium policies: $\eQBFpol{x}{\FAUX}{\eQBFpol{z}{\FAUX}{\lambda}}$ and  $\eQBFpol{x}{\VRAI}{\eQBFpol{z}{\VRAI}{\lambda}}$. The reader can verify that the formula $\exists z \exists x \varphi$ has the same equilibrium policies. This result is expected since, as in QBFs two quantifiers in the same block (therefore, with the same type) can be reversed without affecting the satisfiability. 
\item $\forall x \exists z \varphi$ has the equilibrium policy $\uQBFpol{x}{\eQBFpol{z}{\FAUX}{\lambda}}{\eQBFpol{z}{\VRAI}{\lambda}}$. In this case, the formula  $\forall z \exists x \varphi$ has the same equilibrium policies. However, as in QBFs, this is not generally true.
\item Finally, the reader can verify that the formulas $\exists x \forall z \varphi$, $\exists z \forall x \varphi$, $\forall z \forall x \varphi$ and $\forall x \forall z \varphi$ have no equilibrium policies.
\end{enumerate}
\end{example}
\begin{sloppypar}
Algorithm~\ref{alg:QEM} computes all equilibrium policies associated to an equilibrium model. 
We provide some comments on our algorithm.
The base cases, $\exists x\; \varphi$ and $\forall x\; \varphi$  consider only total $\GTROIS$ interpretations that satisfy $\varphi$ and cannot be discarded by (only) minimising the truth of $x$. 
Also, non-total interpretations that satisfy $\varphi$ will provide no equilibrium policy at all but they will be kept since they can be used to disregard other interpretations.  
For the inductive step, the operator \MCEXISTS\ (resp. \MCFORALL) applies the minimisation condition (Definition \ref{def:mcexists}.$(T_{\exists x})$ (resp. Definition \ref{def:mcforall}.$(T_{\forall x})$)) in order to discard equilibrium policies that are not minimal with respect to Definition~\ref{def:EM}. 
Also potential interpretations that can be used to discard total models are kept for further tests.
For sake of regularity, in $T_{\UNDEMI}$, QBF policies are built from non-crisp models but, of course, all those models can be put into the set of non-crisp models $HC_{\UNDEMI}$ and $T_{\UNDEMI}$ left empty.
For sake of simplicity, the set of non-crisp models is not managed optimally but the relation $\trianglelefteq$ may be used to keep only the minimal elements.
A \texttt{Prolog} prototype is available at \url{https://leria-info.univ-angers.fr/~igor.stephan/Research/EQUILIBRIUM_POLICIES/equilibrium_policies.html}.
\end{sloppypar}

\section{Conclusions and Future Work}\label{sec:conclusions}
In this paper, we considered Quantified Answer Set Programming (QASP) as defined in~\cite{ua8046,falaroscso21a}, where propositional quantification is applied over propositional variables. We demonstrated that the two referenced semantics are not equivalent.
We also addressed the problem of policy generation~\cite{cosfarlanlebmar06} under the semantics of~\cite{falaroscso21a}, introducing an algorithm to compute all minimal policies associated with a QBF. Our approach relies on quantified propositional \GTROIS{}~\cite{BaazCZ00} as a monotonic base and uses ideas from equilibrium logic for minimality. The algorithm has been implemented in a \texttt{Prolog} prototype, allowing practical evaluation on real examples.
For future work, we plan to explore the following directions:
\begin{enumerate}[itemsep=0pt]
    \item \textbf{Stéphan's semantics:} The interpretation of propositional quantification as filters over stable models~\cite{falaroscso21a} aligns with ASP, but restricts players from explicitly choosing truth assignments. Stéphan’s approach allows such moves; we aim to further explore its potential in Knowledge Representation.
    \item \textbf{Algorithmic improvements:} Our current algorithm exhaustively explores all cases, including potentially redundant ones. We intend to first explore the complexity of the algorithm, which we did not consider in this paper, as well as to study the formal properties of our framework to simplify and optimize the procedure.
\item \textbf{Syntactic subfragments of QBF:} Given that QASP satisfiability is PSPACE-complete~\cite{falaroscso21a}, we are interested in identifying syntactic restrictions that may yield lower complexity, following insights from~\cite{StephanM09,abs-2411-10093}.
    \item \textbf{Application to two-player games:} QASP has been applied to model two-player games~\cite{Saffidine}. We plan to apply our policy-extraction algorithm in this context, moving beyond satisfiability checking to identify concrete strategies.
\end{enumerate}


\end{document}